\documentclass[12pt]{iopart}

\usepackage{iopams}
\newtheorem{dfn}{Definition}

\newtheorem{rem}{Remark}
\newtheorem{prop}{Proposition}
  
\begin{document}

\title[ Lie-Hamilton System and QHJ equation ]{On some signatures of Lie-Hamilton System in Quantum Hamilton Jacobi Equation }

\author{Arindam Chakraborty}

\address{Physics Department,  Heritage Institute Of Technology\\
	Chowbaga Road, Anandapur, Mundapara, Kolkata-700107
	Tel.: 03366270502\\
Kolkata, West Bengal,
India\\
}
\ead{arindam.chakraborty@heritageit.edu}
\vspace{10pt}
\begin{indented}
\item[]February 2026
\end{indented}

\begin{abstract}
The general forms of Quantum Hamilton Jacobi Equation for a particle of constant mass,  position dependent effective mass and non-Hermitian Effective mass Swanson model have been considered. It has been found that the said equations can be recast in the form of Cayley-Klein Riccati equations which admit a Lie-Hamilton structure. The possible expressions of Lie symmetry and Lie Integral have  also been considered.
\end{abstract}

%

\section{Introduction}
The present work attempts to connect \textbf{Quantum Hamilton Jacobi (QHJ) equation} under various contexts with the so called \textbf{Lie Hamiltonian System (LHS)}, the two hitherto uncorrelated domains, via \textbf{Cayley-Klein Riccati (CKR) equation}. QHJ method \cite{kapoor22} has been developed as one of the consequences of numerous attempts to show quantum theory as an extension of theories of Classical Mechanics and provides an alternative method of finding values of measurable quantities independently of Schr\"odinger's spectral problem. Though the initial work on QHJ theory was envisaged by Dirac, Jordan, Schwinger and many other with notions like \textbf{Quantum Action} and \textbf{Quantum Momentum Function (QMF)}, the relatively recent comprehensive study on QHJ method by employing \textbf{action-angle technique} and complex analysis was suggested by Leacock and Padgett \cite{leacock83, padgett83}. From then on, with its own share of successes and limitations like many other formalisms \cite{daniel01} of quantum mechanics QHJ method has become able to engender a distinct domain of enquiry by solving problems starting from prototypical one like harmonic oscillator, particle-in-a-box \cite{kapoor22} Coulomb scattering \cite{rowe85} to periodic potential \cite{ranjani04, ranjani05, kapoor22}, position dependent effective mass system \cite{yes10} super-symmetric system \cite{bal01, ranjani13}, multi-particle system \cite{djama18}, quantum vortices \cite{chou08} and even to the recent surge of interest in non-hermitian \cite{yes07} and $\mathcal{PT}$-symmetric systems \cite{yes07, ranjani2005} and many other issues of interest \cite{bracken20, girard22, ture25}. 

On the other hand, the study of \textbf{Lie Systems}, often referred to as \textbf{Lie-Scheffers Systems}, as first order non-autonomous systems of differential equations dates back to the closing decades of 19th century when the proof of the so called Lie-Scheffers theorem elaborated the criterion of  local superposition rule in terms of \textbf{Vessiot Guldberg Lie Algebra (VGLA)} \cite{delucas20}. Among many such Lie Systems the so called \textbf{Lie-Hamilton Systems} are of particular interest since for each of them the corresponding VGLA is composed of \textbf{Hamiltonian Vector Fields} with respect to a \textbf{Poisson Structure} equipped with atleast a local \textbf{Symplectic Form} \cite{west78, vais94}. This fact imparts a geometrical meaning to the system in terms of Hamiltonian Functions obtained from the contraction of relevant symplectic form by VGLA vector fields in one hand and relates each of the vector fields  with its respective Hamiltonian function via \textbf{Poisson Bi-vector} on the other. The said Hamiltonian functions themselves constitute a Lie algebra with respect to the Poisson Bracket induced by the related symplectic form and this algebra is isomorphic to the VGLA under consideration. Considerable attention has been given through a number of studies on LHS like Smorodinsky-Winternitz oscillator \cite{wint67}, Kummer-Schwarz equation \cite{berk07, delucas12}, Ermakov system \cite{cari08, delucas11}, Lotka-Volterra population dynamics \cite{jin05} etc. Several recent advancements concerning superposition rule \cite{ball13}, planner Riccati Equation \cite{ball15}, Poisson-Hopf algebra deformation \cite{ball18},  contact Lie system \cite{delucas23}, Schwarzian systems \cite{cari13}, $k$-symplectic systems \cite{vila15}, multi-symplectic systems \cite{gracia22} have already been done regarding the theoretical development of the subject and its close connection with geometric aspects of differential equations and geometric mechanics.

The present article earns its motivation from two recent pioneering works, one by Kapoor et. al. \cite{kapoor22} regarding QHJ theory and the other one from de Lucas et. al. \cite{delucas20} regarding LHS. We want to observe the so called QHJ equation purely from the perspective of \textbf{Non-Linear Ordinary Differential Equation (NLODE)} where the QMF resembles  an NLODE with complex coefficients in general and hence its solution becomes a complex function of position $x$. We consider three such cases : (I) a single particle system with constant mass in an $x$-dependent potential, (II) a system with position dependent effective mass in an $x$-dependent potential and (III) non-Hermitian effective mass Swanson model. It has been shown that in each case it is possible to write the corresponding QHJ equation 
as a system of coupled NLODE's involving real dependent and independent variables. This system of equations resembles the so called CKR system of equations in a plane. Our present purpose is not to investigate energy eigenvalues and other quantum mechanical aspects but to identify inherent  geometrical structure of systems of  Lie-Hamilton type and aspects like \textbf{Lie Symmetry} or \textbf{Lie Integral}. In section-\ref{constmass} we construct the CKR system for each of the above mentioned cases to identify the $x$ dependent coefficients $\{a_j : j=1, 2, 3\}$. In section-\ref{lh} we briefly discuss the basics of LHS and the way present systems can be understood in terms of some basic signatures of LHS like Poisson structure, symplectic form, Hamiltonian functions and Poisson bi-vector etc. This section concludes with the construction of possible \textbf{Lie Symmetry}  and \textbf{Lie Integral} \cite{espinoza11} under special conditions for each of the above three cases.

\section{ Quantum Hamilton Jacobi (QHJ) Equation and Cayley-Klein Riccati (CKR) System}\label{constmass}
\textbf{Case-I : Constant Mass}

In conventional quantum mechanics of a particle of constant mass moving in a potential $V(x)$, the spectral problem is given by the time-independent Schr\"odinger equation
\begin{eqnarray}\label{quantum1}
H\Psi=E\Psi\:\:{\rm{or}}\:\:(T+V(x))\Psi=E\Psi\nonumber\\
-\frac{\hbar^2}{2m}\frac{d^2\Psi(x)}{dx^2}+V(x)\Psi(x)=E\Psi(x)
\end{eqnarray}
where $T$ is the kinetic energy operator and $E$ being the eigenvalue. In Quantum Hamilton-Jacobi (QHJ) framework the so called \textbf{quantum action} $S\dot{=}S(x, E)$ is defined \cite{kapoor22} as
\begin{equation}
\Psi=e^{iS(x, E)/\hbar},\:\:\:\:\:\:\:\:\:\:S(x, E)=-i\hbar\ln\Psi,\:\:\:\:\:\:\:\:i=\sqrt{-1}
\end{equation}
Therefore in terms of $S$, equation-\ref{quantum1} becomes
\begin{equation}\label{quantum2}
\frac{1}{2m}\left(\frac{dS}{dx}\right)^2-\frac{i\hbar}{2m}\frac{d^2S}{dx^2}+V(x)-E=0
\end{equation}
It is to be noted that as $\hbar\rightarrow 0$, equation-\ref{quantum2} becomes the space part of the separated Hamilton-Jacobi equation usually mentioned in classical mechanics \cite{moretti23}. Equation-\ref{quantum2} is called \textbf{Quantum Hamilton Jacobi (QHJ}) Equation viewed as the quantum deformation of the classical Hamilton Jacobi Equation. The quantity 
\begin{equation}
\wp(x, E)=\frac{dS}{dx}=\frac{-i\hbar}{\Psi}\frac{d\Psi}{ dx}
\end{equation}
is defined as \textbf{Quantum Momentum Function (QMF)} which is a complex function for obvious reason and in view of equation-\ref{quantum2} it satisfies the following \textbf{Riccati Equation}
\begin{equation}\label{qhj}
\frac{\wp^2}{2m}-\frac{i\hbar}{2m}\wp^\prime+V(x)-E=0. 
\end{equation}
A single 'prime' over any variable implies first order derivative with respect to $x$.

In conventional QHJ method QMF is considered a complex function and expanded in Laurrent series. The residue at a moving pole is substituted in the action integral to achieve energy quantization. The detail calculations are given in \cite{kapoor22}

In the present discussion instead of quantization we  investigate an underlying Lie Hamiltonian structure corresponding to a system involving a single particle of constant mass $m$. Inorder to make things amenable to CKR  equation  we employ the following transformation 
\begin{eqnarray}
p(x)=\frac{\alpha(x)\wp(x)+\beta(x)}{\gamma(x)\wp(x)+\delta(x)}\:\:\:{\rm{with}}\:\;\:\alpha\delta-\beta\gamma\neq 0
\end{eqnarray}
and choose $\gamma=0$ and $\beta(x)=i\sigma(x)$. Here, $\alpha, \delta, \sigma$ are all real functions. Considering $\hbar=1$, equation-\ref{qhj} becomes
\begin{equation}\label{ckr}
-ip^\prime-a_1(x)+ia_2(x)p+a_3(x) p^2=0
\end{equation} 
where,
\begin{eqnarray}
a_1(x)=W(\sigma, \alpha)(\alpha\delta)^{-1}-\alpha\delta^{-1}f(x)+\sigma^2(\alpha\delta)^{-1}\nonumber\\
a_2(x)=-(2\sigma\alpha^{-1}+W(\delta, \alpha)(\delta\alpha)^{-1})\nonumber\\
a_3(x)=1/\alpha(x)\nonumber\\
f(x)=2m(V(x)-E)
\end{eqnarray}
are real functions of $x$ and $W(u, v)=u^{\prime}v-uv^{\prime}$. Letting $p=p_1+ip_2$, and equating the real and imaginary parts of equation-\ref{ckr} we obtain the \textbf{Cayley-Klein Riccati}  (CKR) \cite{delucas20} system of equations for constant mass given by
\begin{eqnarray}\label{ckrcm}
p_1^\prime = a_2(x)p_1+2a_3(x)p_1p_2\nonumber\\
p_2^\prime=a_1(x)+a_2(x)p_2+a_3(x)(p_2^2-p_1^2)
\end{eqnarray}

\begin{rem}
Equation-\ref{ckrcm} cannot be obtained for a complex potential of the form $V(x)=V_1+iV_2$, where, $V_1$ and $V_2$ are the real functions and consequently the Hamiltonian becomes non-Hermitian. We shall however consider a typical non-Hermitian Hamiltonian in \textbf{Case-III}.
\end{rem}

\vspace{1cm}

\textbf{Case-II : Position Dependent Effective  Mass}

 For a particle whose mass is depending upon the position we consider \textbf{Quantum Effective Mass Hamilton Jacobi (QEMHJ) equation }. \textbf{Psition Dependent Mass Schr\"odinger (PDMS)} equation has appeared in the conventional discourse of quantum mechanics since the work of von Roos \cite{roos83} in the context of semiconductor theory with more than one  suggestions regarding the form of Schr\"odinger equation in Cartesian coordinates. Later, the same idea was persuaded in curved spaces \cite{quesne04} and applied in the contexts of hard-core potential \cite{dong05}, non-Hermitian potentials \cite{jiang05}, isospectrality \cite{mustafa08}, Quantum Hamilton-Jacobi theory \cite{yes10}, spin-orbit interaction \cite{schmidt15}, equivalence/duality between self-adjoint and non-selfadjoint Hamiltonian \cite{rego16}, quasi-exact solvability and super-symmetry \cite{quesne18}, exponential type mass dependence \cite{dong22}, complex mass \cite{sarathi25} etc.

Following von Roos \cite{roos83} and Ye\c{s}ilta\c{s} \cite{yes10}, the general form of QEMHJ equation for a particle of mass $m\dot{=}m(x)$ in one dimension is given by
\begin{equation}
H\Psi=(T+V(x))\Psi=E\Psi 
\end{equation}
where kinetic energy operator
\begin{equation}
T=-\frac{\hbar^2}{4}\left[m^a(x)\frac{d}{d x}m^b(x)\frac{d}{\partial x}m^c(x)+m^c(x)\frac{d}{d x}m^b(x)\frac{d}{d x}m^a(x)\right]
\end{equation}
along with the constraint $a+b+c=-1$. Taking the constraint into account and letting $m(x)=m_0M(x)$ with $\hbar=2m_0=1$ the relevant Hamiltonian becomes \cite{yes10}
\begin{equation}
H=-\frac{1}{M(x)}\frac{d^2}{dx^2}+\frac{M^{\prime}(x)}{M^2(x)}\frac{d}{dx}+V_{\rm{eff}}
\end{equation}
where, 
\begin{equation}
V_{\rm{eff}}=V(x)+\frac{b+1}{2}\frac{M^{\prime\prime}}{M^2}-(1+b-ac)\frac{(M^\prime)^2}{M^3}
\end{equation}
 
Considering $\Psi(x)$ to be an eigenfunction of $H$ with the eigenvalue $E$, the time-independent Schr\"odinger equation becomes
\begin{eqnarray}\label{eigen1}
-\frac{1}{M(x)}\Psi^{\prime\prime}+\frac{M^{\prime}(x)}{M^2(x)}\Psi^{\prime}+V_{\rm{eff}}\Psi=E\Psi\nonumber\\
{\rm{or}}-\frac{1}{M(x)}\frac{\Psi^{\prime\prime}}{\Psi}+\frac{M^{\prime}(x)}{M^2(x)}\frac{\Psi^{\prime}}{\Psi}+V_{\rm{eff}}-E=0
\end{eqnarray}
Introducing the so called \textbf{Quantum Momentum Function} $\wp(x, E)=-i\frac{\Psi^{\prime}}{\Psi}$, equation-\ref{eigen1} becomes
\begin{eqnarray}\label{qhj1}
-i\wp^{\prime}(x, E)+\wp^2(x, E)+i\frac{M^\prime(x)}{M(x)}\wp(x, E)+M(x)(V_{\rm{eff}}(x)-E)=0
\end{eqnarray}
The above equation represents the \textbf{Quantum Hamilton Jacobi Equation} for a system with spatially varying mass. Writing $\wp=\wp_{\rm{re}}+i\wp_{\rm{im}}=p_1+ip_2$ in equation-\ref{qhj1} and separating the real and imaginary parts we get the following pair of coupled non-linear differential equations
\begin{eqnarray}\label{ckrvm}
{\rm{(a)}}\:\: p_1^\prime=2a_3(x)p_1p_2+a_2(x)p_1\nonumber\\
{\rm{(b)}}\:\: p_2^\prime=a_1(x)+a_2(x)p_2+a_3(x)(p_2^2-p_1^2)
\end{eqnarray}
Where, comparing equation-\ref{ckrvm} with equation-\ref{ckrcm}
\begin{eqnarray}\label{coeff1}
a_1(x)=M(x)(E-V_{\rm{eff}})\nonumber\\
a_2(x)=\frac{M^\prime(x)}{M(x)}\nonumber\\
a_3(x)=1
\end{eqnarray}

Equation-\ref{ckrvm} can be considered as \textbf{Cayley-Klein Riccati} system of equations  for system with variable mass {\footnote{Case-I and Case-II shows a kind of equivalence between a system of constant mass and that of effective mass in so far as CKR equation is concerned}}.

\vspace{1cm}

 {\textbf{Case-III : Non-Hermitian Effective Mass Swanson Model}}

The well studied Non-Hermitian Swanson Hamiltonian \cite{yes13, bagchi15, fring16, moha18, rebo22} is given by
\begin{eqnarray}\label{swanson}
H=\nu_0(a^\dagger a+1/2)+\nu_1 a^2+\nu_2 (a^\dagger)^2+\nu_3a+\nu_4a^\dagger
\end{eqnarray}
For spatially varying mass, it is suggestive to consider the generalized ladder operators $\{a, a^\dagger\}$ defined as \cite{bagchi08}
\begin{eqnarray}
a=\alpha_1(x)\frac{d}{dx}+\alpha_2(x)\:\:\:{\rm{and}}\:\:\:
a^\dagger = -\alpha_1(x)\frac{d}{dx}+\alpha_2(x)-\frac{d\alpha_1(x)}{dx}
\end{eqnarray}
respectively, $\{\nu_j : j=0, 1, 2 ,3\}\in \mathbb{R}$ and $\alpha_1, \alpha_2$ are real functions.
Using the above result the Hamiltonian in equation-\ref{swanson} becomes
\begin{equation}
H=-\tilde{\nu}_0\frac{d}{dx}\alpha_1^2\frac{d}{dx}+k_1(x)\frac{d}{dx}+k_2(x)
\end{equation}
where,
\begin{eqnarray}
\tilde{\nu}_0=\nu_0-\nu_1-\nu_2\nonumber\\
k_1(x)=(\nu_1-\nu_2)\alpha_1(x)(2\alpha_2(x)-\alpha_1^\prime(x))+(\nu_3-\nu_4)\alpha_1(x)\nonumber\\
k_2(x)=(\nu_0+\nu_1+\nu_2)\alpha_2^2(x)-(\nu_0+2\nu_2)\alpha_1^\prime(x)\alpha_2(x)\nonumber\\
-(\nu_0-\nu_1+\nu_2)\alpha_1(x)\alpha_2^\prime(x)+\nu_2(\alpha_1(x)\alpha_1^\prime+(\alpha_1^\prime)^2(x))\nonumber\\
+(\nu_3+\nu_4)\alpha_2(x)-\nu_4\alpha_1^\prime(x)+\frac{\nu_0}{2}
\end{eqnarray}
Using the same definition of QMF in Schr\"odinger equation $H\Psi={E}\Psi$ we get
\begin{equation}
-i\wp^\prime(x)-a_1(x)+ia_2(x)\wp+a_3(x)\wp^2=0
\end{equation}
with
\begin{eqnarray}
a_1(x)=\frac{E-k_2}{\tilde{\nu_0}\alpha_1^2}\nonumber\\
a_2(x)=\frac{k_1-2\tilde{\nu}_0\alpha_1\alpha_1^\prime}{\tilde{\nu_0}\alpha_1^2}\nonumber\\
a_3(x)=1
\end{eqnarray}
leading to the similar system of CKR equations as mentioned in equation-\ref{ckrvm}.

\section{\textbf{ Lie-Hamilton System and QHJ equations}}\label{lh}
The idea of Lie system is closely related to the existence of so called superposition rule for a system of first order ordinary differential equations. Let $\mathcal{M}$ be a real manifold with ${\rm{dim}}\:\: \mathcal{M}=n$ and a local coordinate system $\{p_1\cdots p_n\}$. 
\begin{dfn}
A  system of first order ordinary differential equations on a manifold $\mathcal{M}$ of the form
\begin{eqnarray}\label{sup1}
\frac{dp_i}{dx}=X_i(x, \mathbf{p} ),\:\:\:\:\: i=1\cdots n
\end{eqnarray}
is said to admit a \textbf{global superposition principle} if there exists an $x$-independent map $\Omega : \mathcal{M}^m\times \mathcal{M}\rightarrow \mathcal{M}$ such that
\begin{eqnarray}
p(x)=\Omega(p_1(x)\cdots p_m(x); \alpha)
\end{eqnarray}
where, $p_1(x)\cdots p_m(x)$ is a generic family of particular solutions to $X$ and $\alpha=\alpha_1\cdots, \alpha_n\in N$ are constants that appear due to  initial conditions of $X$.
\end{dfn}
 The so called \textbf{Lie-Scheiffer Theorem} ensures that such superposition principle exists iff
\begin{equation}\label{sup2}
X(x, p)=\sum_{\mu=1}^ra_\mu(x)\chi_\mu(p)
\end{equation} 
where, the vector fields $\{\chi_\mu : \mu=1,\cdots, r\}$ spanning an $r$-dimensional Lie algebra called \textbf{Vessiot-Guldberg Lie Algebra (VGLA)} of $X$ denoted by $V^X$.
\begin{dfn}
A \textbf{Lie System} is a system of equations mentioned in equation-\ref{sup1} along with the condition-\ref{sup2}. A Lie system isomorphic to a Lie algebra $\mathfrak{g}$ is called $\mathfrak{g}$-Lie System.
\end{dfn}

In this section we demonstrate an $\mathfrak{sl}(2, \mathbb{R})$-Lie system is possible to be constructed out of equation-\ref{ckrcm} and equation-\ref{ckrvm}.
 Let us define the following vector fields 
\begin{equation}
\chi_1=\partial_{p_2},\:\:\:\:\chi_2=p_1\partial_{p_1}+p_2\partial_{p_2}\:\:\:{\rm{and}}\:\:\:\chi_3=2p_1p_2\partial_{p_1}+(p_2^2-p_1^2)\partial_{p_2}
\end{equation}
as generators of  where, $\partial_{p_j}=\frac{\partial}{\partial p_j}$. It is easy to verify the following commutation relations :
\begin{equation}
[\chi_1, \chi_2]=\chi_1,\:\:\:[\chi_2, \chi_3]=\chi_3\:\:\:{\rm{and}}\:\:[\chi_1, \chi_3]=2\chi_2
\end{equation}
The above algebra is isomorphic to $\mathfrak{sl}(2, \mathbb{R})$ for obvious reason. It is to be mentioned that Lie-Scheffers theorem relates the dimension of the VGLA with the number of particular solutions of the superposition rule and the 
dimension of the manifold in which the Lie system is defined.

\subsection{\textbf{Poisson Structure and Lie-Hamilton System}}
\begin{dfn}
A pair $(\mathcal{M}, \omega)$ is called a \textbf{symplectic manifold} if $\omega$ is a non-degenerate closed $2$-form defined over the manifold $\mathcal{M}$
\end{dfn}

\begin{prop}
The vector fields $\{\chi_i : i=1, 2, 3\}$ called \textbf{symplectic vector fields} admit a family of \textbf{Hamiltonian functions} $\{\mathcal{H}_1, \mathcal{H}_2, \mathcal{H}_3\}=\{-p_1^{-1}, -p_2p_1^{-1}, -(p_1^2+p_2^2)p_1^{-1}\}$ relative to the symplectic form $\omega=p_1^{-2}dp_2\wedge dp_1$. The Poisson bracket $\{\mathcal{H}_i, \mathcal{H}_j\}_\omega=\omega(\chi_i, \chi_j)$ defines a \textbf{Lie-Hamilton Algebra} isomorphic to $\mathfrak{sl}(2, \mathbb{R})$
\end{prop}
\textbf{Proof :}
In order to construct the relevant Hamiltonian functions $\{\mathcal{H}_j : j=1, 2, 3\}$ we consider the $2$-form $\omega=F(p_1, p_2)dp_2\wedge dp_1$. For a VGLA the local Hamiltonians relative to $\omega$, demands the conditions that the Lie derivatives $\mathcal{L}_{\chi_i}\omega=0\:\:\forall\:\:i$. Using the homotopy relation \cite{west78} 
\begin{equation}
\mathcal{L}_{\chi_i}\omega=d\iota_{\chi_i}\omega+\iota_{\chi_i}d\omega\:\:{\rm{with}}\:\:\iota_{\chi_i}(\theta\wedge\alpha)=(\iota_{\chi_i}\theta)\wedge\alpha+(-1)^\nu\theta\wedge(\iota_{\chi_i}\alpha)
\end{equation}
where, $\nu$ is the index of the form $\theta$ and $\iota_{\chi_i}$ is the so called contraction and applying the condition $\mathcal{L}_{\chi_i}\omega=0$ we find
 \begin{eqnarray}
 \partial_{p_2}F=0, \:\:p_1\partial_{p_1}F+p_2\partial_{p_2}F+2F=0\:\:{\rm{and}}\:\nonumber\\2p_1p_2\partial_{p_1}F+(p_2^2-p_1^2)\partial_{p_2}F+4p_2F=0
 \end{eqnarray}
The first two equations give $F(p_1, p_2)=p_1^{-2}$ and hence $\omega=p_1^{-2}dp_2\wedge dp_1$ which is a closed non-degenerate $2$-form on $\mathbb{R}^2_{p_1\neq 0}=\{p_1, p_2\}\in\mathbb{R}^2(p_1\neq 0)$.

The \textbf{Hamiltonian functions} $\{\mathcal{H}_i\in C^\infty(\mathcal{M}) : i=1, 2, 3\}$, $C^\infty(\mathcal{M})$ being the function ring, are given by the relations $\{d\mathcal{H}_i=\iota_{\chi_i}\omega : i=1, 2, 3\}$. Hence, with the help of \textbf{Hamiltonian Vector Fields}, $\{\chi_i : i
=1, 2, 3\}$ we get
\begin{eqnarray}
\mathcal{H}_1=-p_1^{-1},\:\:\mathcal{H}_2=-p_2p_1^{-1},\:\:\: \mathcal{H}_3=-(p_1^2+p_2^2)p_1^{-1}
\end{eqnarray}

The Poisson Bracket relative to $\omega$ as
\begin{eqnarray}
\{\mathcal{H}_i, \mathcal{H}_j\}_\omega=\omega(\chi_i, \chi_j)=p_1^{-2}dp_2\wedge dp_1(\chi_i, \chi_j)\nonumber\\
=p_1^{-2}\{dp_2(\chi_i)dp_1(\chi_j)-dp_1(\chi_i)dp_2(\chi_j)\}
\end{eqnarray}

This gives the following relations of \textbf{Lie Hamilton algebra} $\{\mathcal{H}_1, \mathcal{H}_2, \mathcal{H}_3 : \{\cdot, \cdot\}_\omega\}$ isomorphic to $\mathfrak{sl}(2, \mathbb{R})$
\begin{eqnarray*}
\{\mathcal{H}_1, \mathcal{H}_2\}_\omega=-\mathcal{H}_1, \:\:\{\mathcal{H}_2, \mathcal{H}_3\}_\omega=-\mathcal{H}_3, \:\:\{\mathcal{H}_1, \mathcal{H}_3\}_\omega=-2\mathcal{H}_2\:\:\:\square
\end{eqnarray*}

\begin{rem}
It is to be mentioned that the bracket $\{\cdot, \cdot\}_\omega : C^\infty(\mathcal{M})\times C^\infty(\mathcal{M})\rightarrow C^\infty(\mathcal{M})$ defines a \textbf{Poisson algebra} $(C^\infty(\mathcal{M}), \cdot, \{\cdot, \cdot\}_\omega)$ as a natural consequence of the existence of a symplectic manifold (here $(\mathcal{M}, \omega)$). Every symplectic form presumes a Poisson structure but the reverse is not always true.
\end{rem}
\begin{dfn}
A Lie system is said to possess a \textbf{Lie Hamilton Structure (LHS)} if it is equipped with a manifold defined by the respective dependent variables and a space of Hamilton functions that represents a Lie algebra corresponding to a preassigned definition of Poisson bracket. A Lie system with a Lie-Hamiltonian structure defined on it is called a \textbf{Lie-Hamilton System} (LHS). 
\end{dfn}
For example, in the present example $\{\mathbb{R}_{y\neq 0}, \omega, \mathcal{H}=a_1(t)\mathcal{H}_1+a_2(t)\mathcal{H}_2+a_3(t)\mathcal{H}_3\}$ defines a Lie-Hamilton system.

 Now let us consider the following definition.
\begin{dfn}
A \textbf{Poisson Bi-vector} $\Lambda$ belonging to the exterior space $\bigwedge^2 T\mathcal{M}$ is defined as $G(p_1, p_2)\partial_{p_2}\wedge\partial_{p_1}$
\end{dfn}

Considering a Poisson bi-vector $\Lambda=p_1^2\partial_{p_2}\wedge \partial_{p_1}$, one can verify the relations
\begin{equation}
\chi_i=-\Lambda(d\mathcal{H}_i)\:\:\forall\:\:i=1, 2, 3.
\end{equation}

The above discussion justifies the existence of Lie-Hamilton system induced by Poisson bi-vector ($\Lambda$) and related to QHJ equation. The relevant Lie Algebra is formed by the following definition of Poisson Bracket
\begin{equation}
\{U, V\}_\Lambda=\Lambda (dU, dV)=p_1^2[\partial_{p_2}dU \partial_{p_1}dV-\partial_{p_2}dV \partial_{p_1}dU]
\end{equation}
It is straightforward to verify

\begin{eqnarray*}
\{\mathcal{H}_1, \mathcal{H}_2\}_\Lambda=-\mathcal{H}_1, \:\:\{\mathcal{H}_2, \mathcal{H}_3\}_\Lambda=-\mathcal{H}_3, \:\:\{\mathcal{H}_1, \mathcal{H}_3\}_\Lambda=-2\mathcal{H}_2
\end{eqnarray*}
which is also a Lie algebra isomorphic to $\mathfrak{sl}(2, \mathbb{R})$. We shall resume to this kind of PB in the context of \textbf{$x$-dependent Lie Integral}.

\section{Effective Mass, Potential Function,  Lie Symmetry and Lie Integral}
The following discussion considers the role of position dependent functions in the expression of a typical Lie symmetry for a VGLA isomorphic to $\mathfrak{sl}(2, \mathbb{R})$.
\begin{dfn}
Given a Lie system 
\begin{equation}
\frac{dp_i}{dx}=X_i(x, p_1, p_2)\:\: :\:\:i=1, 2 
\end{equation}

accompanied by an autonomization $\widetilde{X}=\partial_x+X(x, p_1, p_2)=\partial_x+\sum_{j=1}^3a_j(x)\chi_j$ of $X(x, p_1, p_2)$, an operator $Y=\lambda_0(x)\partial_x+\sum_{j=1}^3\lambda_j(x)\chi_j$
is said to represent a \textbf{symmetry operator} iff
\begin{equation}\label{symm}
[Y, \widetilde{X}]=\lambda(x)\widetilde{X}
\end{equation}
\end{dfn}

\begin{prop}
Let $\{\chi_j : j=1, 2, 3\}$ be the generators of $\mathfrak{sl}(2, \mathbb{R})$ algebra with the following commutation relations

\begin{equation}
[\chi_1, \chi_2]=\chi_1,\:\:\:[\chi_2, \chi_3]=\chi_3,\:\:\:[\chi_1, \chi_3]=2\chi_2
\end{equation}

Equation-\ref{symm} holds provided
\begin{eqnarray}\label{symm1}
{(\rm{a})}\:\:\frac{d\lambda_0}{dx}=a_0=-\lambda(x)\nonumber\\
{(\rm{b})}\:\:\frac{d\lambda_1}{dx}=\lambda_0\frac{da_1}{dx}+\lambda_1a_2-\lambda_2a_1+a_0a_1\nonumber\\
{(\rm{c})}\:\:\frac{d\lambda_2}{dx}=\lambda_0\frac{da_2}{dx}+2(\lambda_1a_3-\lambda_3a_1)+a_0a_2\nonumber\\
{(\rm{d})}\:\:\frac{d\lambda_3}{dx}=\lambda_0\frac{da_3}{dx}+\lambda_2a_3-\lambda_3a_2+a_0a_3
\end{eqnarray}
where $\lambda=-a_0$.
\end{prop}

\textbf{Proof}: Directly calculating the commutation relation we get
\begin{eqnarray}
[Y, \tilde{X}]=[-(\partial_x\lambda_0)\partial_x+\lambda_0\sum_{j=1}^3(\partial_xa_j)\chi_j-\sum_{j=1}^3(\partial_x\lambda_j)\chi_j\nonumber\\
+\lambda_1(a_2\chi_1+2a_3\chi_2)-\lambda_2(a_1\chi_1-a_3\chi_3)\nonumber\\
-\lambda_3(2a_1\chi_2+a_2\chi_3)-\lambda(\partial_x+\sum_{j=1}^3a_j\chi_j)]\nonumber\\
+\lambda(\partial_x+\sum_{j=1}^3a_j\chi_j)
\end{eqnarray}
Since, $\{\chi_j : j=1, 2, 3\}$ are linearly independent for being generators of the $\mathfrak{sl}(2, \mathbb{R})$ algebra the above conditions follow (all partial derivatives have been replaced by ordinary derivatives for obvious reason).$\:\:\:\:\square$ 

\subsection{\textbf{Calculation of symmetry operator $Y$}}
First, we attempt to find out, for \textbf{Case-II}, the possible expressions of $\{\lambda_j : j=1, 2, 3\}$ from equation-\ref{symm1} by assuming $\lambda_0=a_0$. This immediately gives $\lambda_0=Ae^x$, $A$ being a constant. With this choice of $\lambda_0$ and possible substitutions of $\{a_1, a_2, a_3\}$ from equation-\ref{coeff1}, equation-\ref{symm1}((b) and (d)) can be rewritten as
\begin{eqnarray}
\frac{d}{dx}\frac{\lambda_1}{M}=\frac{1}{M}\frac{d}{dx}(Ae^xM(E-V_{\rm{eff}}))-\lambda_2(E-V_{\rm{eff}})\nonumber\\
\frac{d\lambda_3}{dx}=A e^x+\lambda_2-\lambda_3\frac{M^\prime}{M}
\end{eqnarray}
leading to
\begin{eqnarray}\label{coeff2}
{(\rm{a})}\:\:\lambda_1(x)\nonumber\\
=M(x)\int^x \left[\frac{1}{M(\xi)}\frac{d}{d\xi}(Ae^{\xi}M(\xi)(E-V_{\rm{eff}}(\xi)))-
\lambda_2(\xi)(E-V_{\rm{eff}}(\xi))\right]\nonumber\\
{(\rm{b})}\:\:\lambda_3(x)=\frac{1}{M(x)}\int^x M(\xi)(Ae^{\xi}+\lambda_2(\xi))
\end{eqnarray}
Using the above results in equation-\ref{symm1}(c) we get the following integro-differential equation of $\lambda_2$
\begin{eqnarray}\label{coeff3}
\frac{d\lambda_2}{dx}=\frac{d}{dx}\left(Ae^x\frac{M^\prime(x)}{M(x)}\right)\nonumber\\
+2M(x)\int^x \left[\frac{1}{M(\xi)}\frac{d}{d\xi}(Ae^{\xi}M(\xi)(E-V_{\rm{eff}}(\xi)))-\lambda_2(E-V_{\rm{eff}}(\xi))\right]\nonumber\\
-2(E-V_{\rm{eff}}(x))\int^x M(\xi)(Ae^{\xi}+\lambda_2(\xi))
\end{eqnarray}
or an integral equation
\begin{eqnarray}\label{coeff4}
\lambda_2=Ae^x\frac{M^\prime(x)}{M(x)}\nonumber\\
+\int^x\left[2M(\xi)\int^\xi \left\{\frac{1}{M(\eta)}\frac{d}{d\eta}(Ae^\eta M(\eta)(E-V_{\rm{eff}}(\eta)))-\lambda_2(E-V_{\rm{eff}}(\eta))\right\}\right]\nonumber\\
-\int^x\left[ 2(E-V_{\rm{eff}}(\xi))\int^{\xi} M(\eta)(Ae^\eta+\lambda_2(\eta))\right]
\end{eqnarray}
Equation-\ref{coeff4} can be solved out for $\lambda_2$ for specific expressions of effective mass $M$ and potential function $V_{\rm{eff}}(x)$ while the solutions to equations-\ref{coeff2}((a), (b)) can be found out by quadrature after putting the expression of $\lambda_2$.

\begin{rem}
It is well known that depending upon the potential function especially those responsible for the band structure the effective mass $M$ can be positive, negative or even zero. It is interesting to see that the present symmetry operator $Y$ takes up different expressions depending upon the sign of the effective mass and shows singularity at $M=0$.
\end{rem}

\textbf{Case-II} provides the motivation to calculate symmetry operators for the other two cases. For \textbf{Case-I} retaining $\lambda_0=Ae^x$ and employing the same procedure as above one can write
\begin{eqnarray}
\lambda_1=g(x)\int\frac{1}{g(x)}\left[Ae^x\left\{\frac{W(\sigma, \alpha)-f(x)+\sigma^2}{\alpha\delta}\right\}-\frac{\lambda_2}{\alpha}\right]dx\nonumber\\
\lambda_3=\frac{1}{g(x)}\int g(x)\left[\frac{Ae^x+\lambda_2}{\alpha}\right]dx\nonumber\\
\lambda_2=Ae^x\frac{g^\prime}{g}+2\int\left[\frac{\lambda_1}{\alpha}-\lambda_3\left\{\frac{W(\sigma, \alpha)-f(x)+\sigma^2}{\alpha\delta}\right\}-\frac{\lambda_2}{\alpha}\right]dx
\end{eqnarray}
where, $\frac{g^\prime(x)}{g(x)}=a_2$.

Similar calculation may be repeated for \textbf{Case-III} as well.

\subsection{\textbf{x-dependent Lie Integral}}\label{li}
In this section we  consider \textbf{Lie Integral} for each of the above system under specific choices of relevant functions and constants.
\begin{dfn}
Given a Lie system defined over a manifold $\mathcal{N}$ of dimension $n$ and an LHS $\{\mathcal{N}, \Lambda, \mathcal{H}\}$ a \textbf{Lie Integral} is the solution of the equation
\begin{eqnarray}
\frac{d\Upsilon}{dx}=\{\mathcal{H}, \Upsilon\}_\Lambda
\end{eqnarray}
\end{dfn}
where, for our present purpose $\Upsilon=\sum_{j=1}^3\Upsilon_j(x)\mathcal{H}_j$ and $\mathcal{H}=\sum_{j=1}^3a_j(x)\mathcal{H}_j$ The above equation is known as \textbf{Euler Equation} of the respective Lie Algebra. A constant or so called conserved quantity implies $\{\mathcal{H}, \Upsilon\}_\Lambda=0$.

In the present situation Euler equation gives
\begin{equation}\label{triplet}
\frac{d\Upsilon_1}{dx}=a_2\Upsilon_1-a_1\Upsilon_2,\:\:\:\frac{d\Upsilon_2}{dx}=2(a_3\Upsilon_1-a_1\Upsilon_3),\:\:\:\frac{d\Upsilon_3}{dx}=a_3\Upsilon_2-a_2\Upsilon_3
\end{equation}
Let us consider $\Upsilon_2=0$ in each of the cases and try to solve the triplet of equations-\ref{triplet} for convenient choices of various quantities. 

\textbullet$\:\:\:${ \textbf{Case-I} : Choosing $\alpha=\delta=1$ and $\sigma\dot{=}\sigma(x)$, we get, 
\begin{equation}
a_1=\sigma^\prime-2m(V(x)-E)+\sigma^2\:\:{\rm{and}}\:\:a_2=-2\sigma(x)
\end{equation}
giving 
\begin{equation}
\Upsilon_1=C_-e^{-2\int\sigma(x)dx} \:\:{\rm{and}}\:\:\Upsilon_3=C_+e^{2\int\sigma(x)dx}
\end{equation}
where, $C_{\pm}$ are integration constants.
Hence, the Lie integral becomes

\begin{eqnarray}
\Upsilon=\Upsilon_1\mathcal{H}_1 +\Upsilon_3\mathcal{H}_3=-C_-\frac{e^{-2\int\sigma(x)}}{p_1}-C_+e^{2\int\sigma(x)}\frac{p_1^2+p_2^2}{p_1}.
\end{eqnarray}

\vspace{0.5cm}

\textbullet$\:\:\:$\textbf{Case-II} : 
Similar method can be repeated  with $a_1=M(x)(E-V_{\rm{eff}}(x))$ and $a_2=\frac{M^\prime}{M}$. We get
$\Upsilon_1=B_-M(x)$ and ${\Upsilon_3^{-1}}=\frac{B_+}{M(x)}$. Hence,  the Lie Integral becomes  
\begin{equation}
\Upsilon=B_-M(x)\mathcal{H}_1+B_+\frac{1}{M(x)}\mathcal{H}_3=-B_-\frac{M(x)}{p_1}-B_+\frac{1}{M(x)}\frac{p_1^2+p_2^2}{p_1}
\end{equation}
 
 \vspace{0.5cm}

\textbullet $\:\:\:\:$ \textbf{Case-III} : Choosing $\nu_1=\nu_2$ and $\nu_3=\nu_4$ we get $k_1=0$. This leads to $\Upsilon_1=K_-\alpha_1^{-2}$ and $\Upsilon_3=K_+\alpha_1^2$. Correspondingly the Lie Integral becomes $\Upsilon=-K_-\frac{1}{\alpha_1^2p_1}-K_+\alpha_1^2\frac{p_1^2+p_2^2}{p_1}$

\section{Connection between Lie Integral and QHJ formalism}
It is to be brought into attention that validity of the construction of Lie integral
in subsection-\ref{li} demands a constraint relation $\Upsilon_2=0$ i. e.; $a_3\Upsilon_1-a_1\Upsilon_3=0$ . For \textbf{Case-I} such a restriction implies the following integro-differential equation of $\sigma(x)$.

\begin{eqnarray}
\sigma^\prime-2m(V{(x)}-E)+\sigma^2=C_0e^{-4\int\sigma(x)dx}\:\:{\rm{where}}\:\:C_0=\frac{C_-}{C_+}
\end{eqnarray}
For bound state problem $E$ is a number and $V(x)$ is a given function of $x$. The choice of integration constant may be equated with $2mE$, where $E$ can be obtained via QHJ method or just claiming it to be a constant. For example, if $\sigma(x)$ is small, $C_0e^{-4\int\sigma(x)}\approx C_0-4C_0\int\sigma(x)dx$ and the choice is $C_0=2mE$. The final calculation boils down to the equation
\begin{equation}\label{sigma}
\sigma^\prime-2mV{(x)}+\sigma^2=-4C_0\int\sigma(x)dx
\end{equation}
This means the term $\sigma(x)$  responsible for modifying the QMF can be approximated as the solution of equation-\ref{sigma}.

For \textbf{Case-II}, the same constraint relation becomes
\begin{eqnarray}
B_-M-B_+\frac{a_1}{M}=0\nonumber\\
B_-M=B_+(E-V_{\rm{eff}})\nonumber\\
=B_+E-B_+ reponsible for modifying\nonumber\\
B_0M(x)=E-V_{\rm{eff}}
\end{eqnarray}
Choosing $b+1=0$ and $a=-c$,
\begin{eqnarray}\label{varimass}
B_0M(x)=E-\left[V(x)+a^2\frac{{M^\prime}^2}{M^3}\right]\nonumber\\
B_0M^\prime=-V^\prime(x)-a^2\frac{d}{dx}\frac{{M^\prime}^2}{M^3}
\end{eqnarray}
This means the form of Lie Integral we have mentioned earlier demand the variable mass $M$ to satisfy the differential equation-\ref{varimass}.
  
For \textbf{Case-III}, similar consideration leads to
\begin{eqnarray}
E-k_2(x)=\frac{\tilde{\nu}_0}{\alpha_1^2(x)}\nonumber\\
k_2^\prime-\frac{2\tilde{\nu}_0\alpha_1^\prime}{\alpha_1^3(x)}=0
\end{eqnarray}
It is to be noted that every time the so called consistency condition involve $E$ as a linear term. Converting the condition to differential equation makes it free of energy $E$ which takes up different values for different excited states. The advantage of choosing differential equation over an equation linear in energy is to set proper integration constant to match it out with the energy value.

\section*{Conclusion} The above discussion deals with the existence of a \textbf{Lie-Hamilton System} for \textbf{Quantum Hamilton Jacobi (QHJ)} equation in one dimension in various circumstances like conventional potential problem to variable mass Schr\"odinger problem and even to the case of non-Hermitian Swanson Hamiltonian. The prime objectives of quantum mechanics like finding out eigenvalues or eigen states have not been envisaged here, rather we assume their existence a priori for a bound state problem. We, in turn, explore the QHJ equation from the perspective of differential equation (here \textbf{Cayley-Klein Riccati equation})to glean out signatures of Lie-Hamilton System, signatures which are relevant for geometric understanding of the problem. The differential geometric features like symplectic form, Poisson bi-vector, Lie derivative, Lie Integral, Lie symmetry have been our concern. It has come out at the end that sometimes, if not most of the cases differential geometric features seem to be more general than the precise classification of the physical problem to be supposedly held categories like Classical or Quantum Mechanical. This view is in consonance with various efforts ever made to view Classical Mechanics as the limiting case of Quantum Formalisms. Similar consideration for Lie-Dirac system is under preparation. Extension of the method employed here to higher dimension, for general non-Hermitian potential and for particle with spin may be suggested as some interesting areas of further investigation.

\section*{Acknowledgement} The author wishes to thank his colleague Dr B. Mal for her assistance to prepare the latex version of the manuscript.

\section*{ORCID iDs}
Arindam Chakraborty https://orcid.org/0000-0002-3414-3785

\end{document}